\documentclass{article}

\usepackage{PRIMEarxiv}
\usepackage[numbers]{natbib}
\usepackage[utf8]{inputenc} 
\usepackage[T1]{fontenc}    
\usepackage{hyperref}       
\usepackage{url}            
\usepackage{booktabs}       
\usepackage{amsfonts}       
\usepackage{nicefrac}       
\usepackage{microtype}      
\usepackage{lipsum}
\usepackage{fancyhdr}       
\usepackage{graphicx}       
\graphicspath{{media/}}     

\usepackage{centernot}
\usepackage{multicol}
\usepackage{multirow}
\usepackage{mathtools}
\usepackage{enumitem}
\usepackage{booktabs}
\usepackage{amsmath,amsfonts,amsthm}
\usepackage{algorithm}
\usepackage{algpseudocode}
\usepackage{soul}

\newtheorem{condition}{Condition}
\usepackage{bbm}
\usepackage{xcolor}
\usepackage{pifont}
\newcommand{\cmark}{\ding{51}}%
\newcommand{\xmark}{\ding{55}}%
\usepackage{adjustbox}
\usepackage{pgfplots}
\pgfplotsset{compat=1.17}
\usepackage{textcomp}
\newtheorem{instance}{Example}

\renewcommand{\paragraph}[1]{\smallskip\noindent\textbf{#1}}

\newtheorem{Claim}{Claim}
\newtheorem{claim}{Claim}
\newtheorem*{CLAIM}{Claim}

\newtheorem{theorem}{Theorem}
\newtheorem*{Theorem}{Theorem}
\newtheorem{corollary}{Corollary}
\newtheorem{Corollary}{Corollary}
\newtheorem{definition}{Definition}

\usepgfplotslibrary{external}
\title{Combinatorial Civic Crowdfunding with Budgeted Agents: Welfare Optimality at Equilibrium and Optimal Deviation
\thanks{To appear in the Proceedings of the Thirty-Seventh AAAI Conference on Artificial Intelligence \textbf{(AAAI '23)}. A preliminary version of this paper titled ``\textit{Welfare Optimal Combinatorial Civic Crowdfunding with Budgeted Agents}'' also appeared at \textbf{GAIW@AAMAS '22}.}
}

\author{
  Sankarshan Damle, Manisha Padala, Sujit Gujar \\
  Machine Learning Lab\\
  International Institute of Information Technology (IIIT), Hyderabad \\
  \texttt{\{sankarshan.damle, manisha.padala\}@research.iiit.ac.in} \\
  \texttt{sujit.gujar@iiit.ac.in}
}

\begin{document}
\maketitle

\begin{abstract}
Civic Crowdfunding (CC) uses the ``power of the crowd'' to garner contributions towards public projects. As these projects are non-excludable, agents may prefer to ``free-ride,''  resulting in the project not being funded. For single project CC, researchers propose to provide refunds to incentivize agents to contribute, thereby guaranteeing the project's funding. These funding guarantees are applicable only when agents have an unlimited budget. This work focuses on a combinatorial setting, where multiple projects are available for CC and agents have a limited budget. We study certain specific conditions where funding can be guaranteed. Further, funding the optimal social welfare subset of projects is desirable when every available project cannot be funded due to budget restrictions. We prove the impossibility of achieving optimal welfare at equilibrium for any monotone refund scheme. We then study different heuristics that the agents can use to contribute to the projects in practice. Through simulations, we demonstrate the heuristics' performance as the average-case trade-off between welfare obtained and agent utility.
\end{abstract}

\keywords{Civic Crowdfunding \and Nash Equilibrium \and Social Welfare}

\section{Introduction}
\label{sec:intro}

Local communities often find it beneficial to elicit contributions
from their members for \emph{public} good projects.
E.g., the construction of markets, playgrounds, 
and libraries, among others~\cite{london}. These goods provide 
the local community with social amenities, generating social welfare.
This process of {generating} funds from members towards 
community services is referred to as  \emph{Civic Crowdfunding} (CC). 
CC is instrumental in changing the interaction between local 
governments and communities. 
It empowers citizens by allowing participation in the design and planning of public good projects~\cite{van2021civic}. 
Such democratization of public projects has led CC to become an active area of research~\cite{diederich2016group,goodspeed2017community,chandra2017referral,wang2020mechanism,yan2021optimal,cason2021early}. Moreover, introduction of web-based CC platforms~\cite{wiki:kickstarter,spacehive} has added to its popularity.

As depicted in Figure~\ref{fig:my_motivation}, typically, multiple projects are simultaneously available for CC. We refer to CC for multiple projects as \emph{combinatorial} CC. Formally, CC comprises strategic agents who observe their valuations for the available public projects. Each project has a known target cost and deadline. 
Each agent contributes to the available projects as per its valuations and within its \emph{budget}.
The agent valuations are such that the overall sum is greater than the project's target cost, i.e., there is enough valuation (interest) for the project's funding. 
A project is \emph{funded} when the agents' total contribution meets the target cost within the deadline. When funded, each agent obtains a quasi-linear utility equivalent to its valuation for the project minus its contribution. In turn, the community generates social welfare -- the difference in the project's total valuation and cost.

\paragraph{Free-riding.} The primary challenge in CC is due to the non-excludability of the public projects. That is, the citizens can avail a project's benefit without contributing to its funding.
Consequently, strategic agents may free-ride and merely wait for others to fund the project.
When the majority decides to free-ride, the project remains unfunded despite sufficient interest in its funding~\cite{stroup2000free}. 
To persuade strategic agents to contribute, researchers propose to provide additional incentives to them in the form of \emph{refunds}.

\begin{figure}[!t]
    \centering
    \includegraphics[width=0.7\linewidth,trim={250pt 150pt 250pt 320pt},clip]{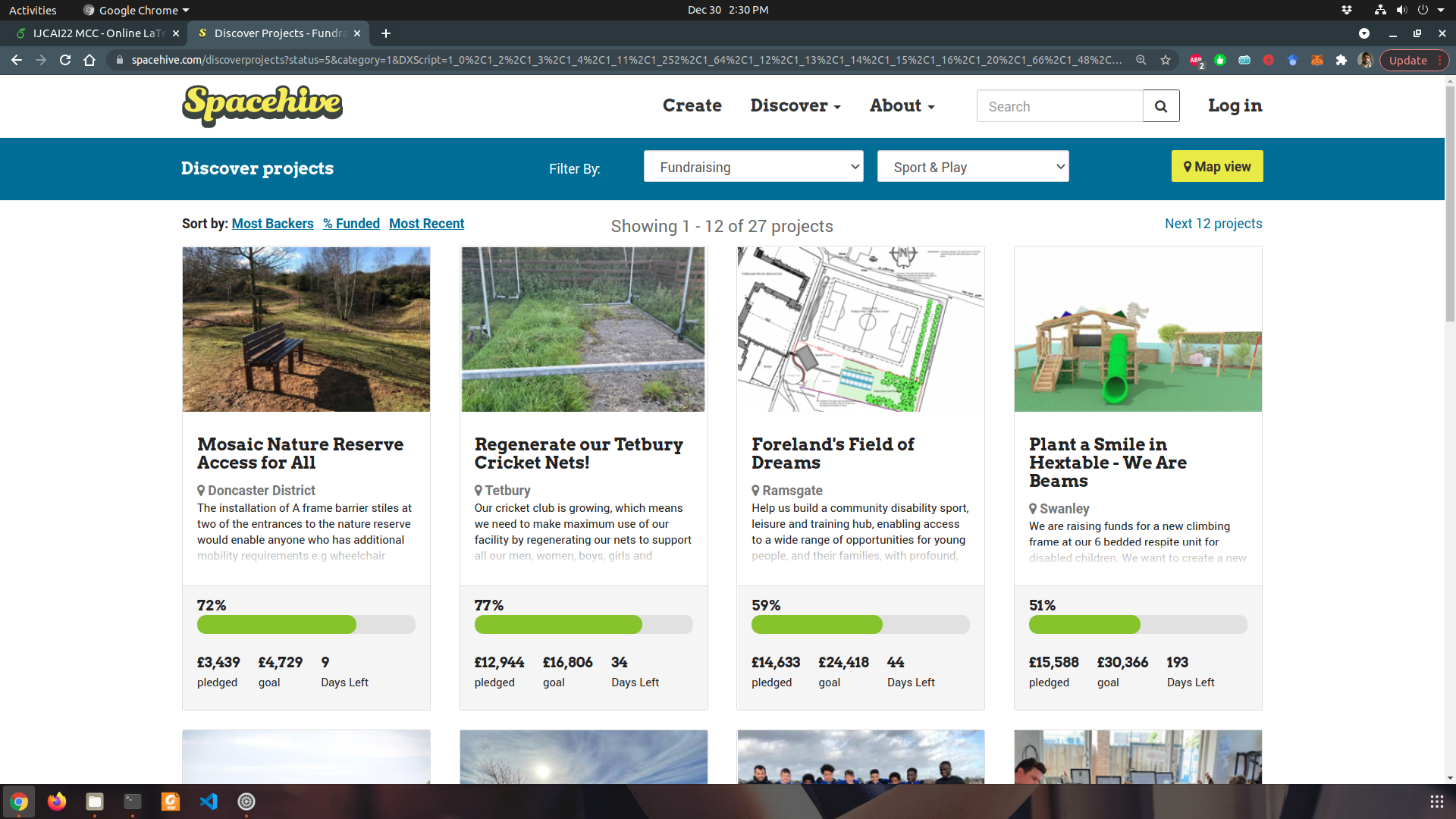}
    \caption{Example instance of Combinatorial Civic Crowdfunding (CC)~\cite{spacehive}. 
    }
    \label{fig:my_motivation}
\end{figure}
%

\paragraph{Refunds.} \citet{zubrickas2014provision} presents PPR, \textit{Provision Point Mechanism with Refunds}, which employs the first such refund scheme. PPR assumes that a \emph{central planner} keeps some refund budget aside. If the project is not funded, the planner
returns the agent's contribution and an additional refund proportional to the agent's contribution. The refund scheme incentivizes the agents to increase their contribution to obtain a greater refund. The characteristic of the resulting game is that the public project is funded at equilibrium. However, PPR and subsequent works (\cite{damlePricai} and references cited therein) assume that agents have an unlimited individual budget; in reality, the agents may have a limited budget. To this end, we aim to determine the funding guarantees for a subset of public projects that maximize social welfare within the available budget.

\subsection*{Our Approach and Contributions}
This paper lays the theoretical foundation for combinatorial CC with budgeted agents. Table \ref{tab:my_label} presents the overview of our results, described in detail next.

\smallskip\noindent\textbf{Budget Surplus (BS).} We first study the seemingly straightforward case of \emph{Budget Surplus}, i.e., the overall budget across the agents is more than the projects' total cost. For this, it is welfare optimal to fund all the projects. Despite the surplus budget, we show that the projects' funding cannot be guaranteed at equilibrium~(Theorem~\ref{thm::noneqm} and Corollary~\ref{cor::bs}).

\smallskip\noindent\textbf{Subset Feasibility (SF).} We observe that the budget distribution among the agents plays a significant role in deciding the funding status of the projects. Conditioning on the budget distribution, we introduce \emph{Subset Feasibility} of a given subset of projects. We prove that Subset Feasibility coupled with Budget Surplus guarantees funding of every available project at equilibrium (Theorem~\ref{thm:bf_bs}), thereby generating the maximum possible social welfare. 

\smallskip\noindent\textbf{Budget Deficit (BD).} 
Trivially, in the case of \emph{Budget Deficit} -- when there is no Budget Surplus -- one can only fund a subset of projects. It may be desirable that such a subset is welfare-maximizing within the budget. We refer to the funding of the socially welfare optimal subset at equilibrium as \emph{socially efficient equilibrium}. For this case, we present the following results.

First, we show that, in general, achieving socially efficient equilibrium is impossible for any refund scheme (Example~\ref{eg::one}). Next, we prove that even with the stronger assumption of Subset Feasibility, it is still impossible to achieve socially efficient equilibrium  (Theorem \ref{thm:bf_bd}). Specifically, we prove that strategic deviations may exist for agents such that the optimal welfare subset remains unfunded. We then show that it is NP-Hard for an agent to find its optimal deviation, given the contributions of all the other agents~(Theorem~\ref{thm::nphard} and Corollary~\ref{corr:nphard}). Due to Theorem~\ref{thm:bf_bd} and hardness of optimal deviation (Theorem~\ref{thm::nphard}), we construct five heuristics for the agent's contributions and empirically study their social welfare and agent utility through simulations (Section~\ref{sec:expt}).

%
\begin{table}[!t]
    \centering
    \adjustbox{max width=\columnwidth}{
    \begin{tabular}{lc}
    \toprule
    \textbf{Property}     &   \textbf{Socially Efficient Equilibrium}\\
    \midrule
     Budget Surplus       & \xmark ~(Corollary~\ref{cor::bs})  \\
     Budget Surplus + Subset Feasibility    &  \cmark ~(Theorem~\ref{thm:bf_bs})\\
     Budget Deficit + Subset Feasibility  & \xmark ~(Theorem~\ref{thm:bf_bd})\\
     \bottomrule
    \end{tabular}}
    \caption{Overview of our theoretical results. 
    }
    \label{tab:my_label}
\end{table}
%


\section{Related Work}
\label{sec:rw}
Several works study the effect of agents' contribution to public projects~\cite{wang2020mechanism,pp1,pp2,brandl2022funding}. One way of modelling agent contribution is using \emph{Cost Sharing Mechanisms} (CSMs) in CC~\cite{moulin1994serial}. More concretely, CSMs focus on sharing the cost among the strategic agents to ensure that an efficient set of projects are funded \cite{moulin1994serial,moulin2001,tim2018,shahar17,birmpas2019cost}. The authors in \cite{wang2020mechanism,ohseto2000} model CSMs for non-excludable public projects and provide agent contributions that ensure specific desirable properties, e.g., individual rationality and strategy-proofness. However, these works do not guarantee funding at equilibrium since agents are strategic and CSMs do not offer refunds.

In another line of work, funding of public projects is modeled as \emph{Participatory Budgeting} (PB) \cite{aziz2021participatory}. \citet{brandl2022funding} study a model without targets costs and without quasi-linear utilities, applicable for making donations to long-term projects. Generally, in the PB literature, the utility of an agent is determined by the number of projects funded or the costs of the projects~(e.g., \cite{azizadt21,sreedurga2022maxmin}), whereas in CC, it is the difference between the agent's valuation and contribution. \citet{aziz2022coordinating} consider a PB model with budget constraints and quasi-linear utilities when the project is funded. The authors prove the impossibility of finding an approximation to welfare optimal funded subset of projects which also ensure weak participation, i.e., positive utility to all agents. However, they do not assume strategic agents.

For excludable public projects, \citet{pp2} focus on effort allocation by strategic agents towards the project's completion. Contrarily, we focus on funding guarantees of non-excludable public projects with strategic agents.

\paragraph{CC with Refunds.} In the seminal work, \citet{bagnoli1989provision} present Provision Point Mechanism (PPM) for single project CC, without refunds. Consequently, PPM consists of several inefficient equilibria~\cite{bagnoli1989provision,healy2006learning}. Agents may also free-ride since the projects are non-excludable~\cite{stroup2000free}. To overcome such limitations, \citet{zubrickas2014provision} presents PPR, a novel mechanism that offers refunds proportional to contributions.
Based on the attractive properties of PPR, other works propose different refund schemes for different agent models and strategy space~\cite{damlePricai,chandra2016crowdfunding,damle2019ijcai,damle2019aggregating}. However, these works only focus on a single project with agents having unlimited budgets.

Among recent works, \citet{padala21} attempt to learn equilibrium contributions when agents have a limited budget in combinatorial CC using Reinforcement Learning. However, the work does not provide any funding guarantees -- welfare or otherwise -- for the projects. 
 
\citet{pp1} analyze the existence of cooperative Nash Equilibrium for funding a single public project using `external investments.' Their work considers agents to have binary contributions, unlike our setting, where agent contributions are in $\mathbb{R}_+$. Moreover, in their utility structure, agents receive a fixed fraction of the total contribution when the project is funded; otherwise, their contribution is returned. That is, they do not model agent valuations.  

We remark that while our work is motivated by the existing CC literature, it remains fundamentally different. To the best of our knowledge, we are the first to study the funding guarantees for combinatorial CC with budgeted agents. We also focus on an agent's equilibrium behavior and study the hardness of the optimal strategy for the agents. 

\section{Preliminaries}
\label{sec:prelim}
This section presents our CC model and important definitions. We also summarize PPR for the single project case.

\subsection{Combinatorial CC Model}
Let $P=\{1,\ldots,p\}$ be the set of projects to be crowdfunded with target costs $T=\{T_1,\ldots,T_p\}$. Let $N=\{1,\ldots,n\}$ denote the set of agents interested in contributing to all projects. We consider a limited budget for each agent $\gamma = (\gamma_1, \ldots, \gamma_n)$. Each agent $i$ has a private valuation for the project $j$, denoted by $\theta_{ij}\ge0$. We consider \textit{additive} valuations, i.e., an agent $i$ has a value of $\sum_{j\in M} \theta_{ij}$ for a funded subset $M\subseteq P$. Let $\vartheta_j = \sum_{i\in N}\theta_{ij}$ denote the total valuation in the system for the project $j$. An agent $i$ contributes $x_{ij}\in\mathbb{R}_+$ to project $j$, s.t., $\sum_{j\in P} x_{ij} \leq \gamma_i$. The total contribution towards a project $j$ is denoted by $C_j = \sum_{i\in N} x_{ij}$. The project is funded if $C_j \geq T_j$ by the deadline, and each agent gets the funded utility of $\theta_{ij} - x_{ij}$. 
If the project is unfunded $(C_j<T_j)$, the agents are returned their contributions $x_{ij}$ and in some mechanisms, additional refunds, as defined later.

\paragraph{Welfare Optimal.} Ideally, when there is limited budget, it may be desirable to fund welfare optimal subset defined as follows. Note that, the welfare obtained from project $j$ if funded is $\vartheta_j - T_j$ and zero otherwise~\cite{borgers2015introduction,chakrabarty2014welfare}\footnote{All the results presented in this paper also hold if $ P^\star \in \underset{M \subseteq P}{\mbox{arg max}}  \sum_{j\in M} \vartheta_j ~ \mbox{s.t.}~ \sum_{j\in M} T_j \leq \sum_{i\in N} \gamma_i. $ }. 

\begin{definition}[Welfare Optimal]
\label{def:wo}
A set of projects $P^\star\subseteq P$ is welfare optimal if it maximizes social welfare under the available budget, i.e.,
\begin{equation}
P^\star \in \underset{M \subseteq P}{\mbox{arg max}}  \sum_{j\in M} \left(\vartheta_j - T_j\right)~ \mbox{s.t.}~ \sum_{j\in M} T_j \leq \sum_{i\in N} \gamma_i. 
\end{equation}

\end{definition}
\noindent  We make the following observations based on Definition~\ref{def:wo}.
 
 \begin{itemize}[leftmargin=*,noitemsep]
     \item Finding $P^\star$ requires public knowledge of $\vartheta$s, $T$s and the value $\sum_{i\in N} \gamma_i$. Contrary to the PB or CSM literature, the aggregate valuation $\vartheta$ is assumed to be public knowledge in the CC literature~\cite{zubrickas2014provision,chandra2016crowdfunding}. Similarly, we also assume that the overall budget in the system $\sum_{i\in N} \gamma_i$ is public knowledge. This may be done by deriving the overall budget by aggregating citizen interest~\cite{alegre2020case,hudexchange}.
    \item Computing $P^\star$ is NP-Hard as it can be trivially reduced from the KNAPSACK problem. However, note that our primary results focus on $P^\star$'s funding guarantees at equilibrium (and not actually computing it). Moreover, computing $P^\star$ may also not be a deal breaker as the number of simultaneous projects available will not be arbitrarily large. One may also employ FPTAS~\cite{lawler1977fast}.
 \end{itemize}

\paragraph{Refund Scheme.} We define the refund scheme for each project $j\in P$ as  $R_j(B_j, x_{ij}, C_j): \mathbb{R}_{+}^{3} \rightarrow \mathbb{R}_+ $ s.t. $r_{ij}=R_j(B_j, x_{ij}, C_j)$ is agent $i$'s refund share for contributing $x_{ij}$ to project $j$.
The overall budget for the refund bonus $B_{j}>0$ is {public} knowledge. Typically, if a project is unfunded, the agents receive $r_{ij}$, and zero refund otherwise. The total refunds distributed for project $j$ can be such that $\sum_{i\in N} r_{ij}= B_j$ (e.g., \cite{zubrickas2014provision}) or $\sum_{i\in N} r_{ij}< B_j$ (e.g., \cite{chandra2016crowdfunding,damlePricai}). Throughout the paper, we assume that $\sum_{i\in N} r_{ij}= B_j~\forall j$.

The CC literature also assumes that $R$ is anonymous, i.e., refund share is independent of agent identity. Further, consider the following condition for a refund scheme, assuming $R$ is differentiable w.r.t. $x$.

\begin{condition}[Contribution Monotonicity (CM)~\cite{damlePricai}]\label{cond1}
A refund scheme $R(x;\cdot)$ satisfies Contribution Monotonicity (CM) if it is strictly monotonically increasing with respect to the contribution $x$, i.e.,
    $
    \frac{\partial R(x;\cdot)}{\partial x} > 0.
    $
\end{condition}

\subsection{Agent Utilities and Important Definitions} Let $\mathcal{M}_{CC}=\left\langle P,N,\gamma,T,(\vartheta_j)_{j\in P},(R_j)_{j\in P},(B_j)_{j\in P} \right\rangle$ define a general combinatorial CC game. In this, the overall agent utility can be defined as.

\begin{definition}[Agent Utility]\label{def:au}
Given an instance of $\mathcal{M}_{CC}$, with agents having valuations $[\theta_{ij}]$ and contributions $[x_{ij}]$, the utility of an agent $i$ for each project $j\in P$ is given by $\sigma_{ij}(\theta_{ij}, x_{ij}, r_{ij}, C_j, T_j):\mathbb{R}_+^5 \rightarrow \mathbb{R}$ 
\begin{equation*}
\sigma_{ij}(\cdot)= \mathbbm{1}_{C_j \geq T_j} \cdot \underbrace{(\theta_{ij} - x_{ij})}_\text{Funded utility $\sigma_{ij}^F$} + ~\mathbbm{1}_{C_j < T_j} \cdot \underbrace{r_{ij}}_\text{Unfunded utility $\sigma_{ij}^U$}
\end{equation*}
where $\mathbbm{1}_X$ is an indicator variable, such that $\mathbbm{1}_X=1$ if $X$ is true and zero otherwise.
\end{definition}

An agent $i$'s utility for project $j$ is either $\sigma_{ij}^F=\theta_{ij} - x_{ij}$ when $j$ is funded, and $\sigma_{ij}^U=r_{ij}$ otherwise. Let $U_i(\cdot)$ denote the total utility an agent $i$ derives, i.e., $U_i(\cdot)=\sum_{j\in P} \sigma_{ij}$. 
This incentive structure induces a game among the agents. 
As the agents are strategic, each agent aims to provide contributions that maximizes its utility. As such, we focus on contributions which follow pure strategy Nash equilibrium.

 \begin{definition}[Pure Strategy Nash Equilibrium (PSNE)]
 \label{def:psne}
A contribution profile $(x_{i1}^*,\ldots, x_{ip}^*)_{i \in N}$ is said to be Pure Strategy Nash equilibrium (PSNE) if, $\forall i \in N$,
$$ \sum_{j\in P} \sigma_{ij}(x_{ij}^*,x_{-ij}^*;\cdot) \geq \sum_{j\in P} \sigma_{ij}(x_{ij},x_{-ij}^*;\cdot),\  \ \forall x_{ij}.
$$
where $x_{-ij}^*$ is the contribution of all agents except agent $i$.
\end{definition}

\smallskip\noindent\textbf{Efficacy of PSNE Contributions.} PSNE is the standard choice of solution concept in CC literature~\cite{zubrickas2014provision,damlePricai,chandra2016crowdfunding,damle2019ijcai}. \citet{zubrickas2014provision} shows that for an appropriate refund bonus (see Eq.~\ref{eqn::equal}), their PSNE strategies are the unique equilibrium of the mechanism. Moreover, \citet{cason2017enhancing} empirically validate the effectiveness of these PSNE strategies using real-world experiments.

Given the contributions $[x_{ij}^*]$, we can compute the set of the projects that are funded and unfunded at equilibrium. Throughout the paper, we refer to the funding of $P^\star$ at equilibrium as \emph{socially efficient equilibrium}. We next define budget surplus.

\begin{definition}[Budget Surplus (BS)]\label{def:bs}
There is enough overall budget to fund each project $j\in P$, i.e.,
$\sum_{i\in N} \gamma_i \geq \sum_{j\in P} T_j$.
\end{definition}

We refer to the scenario $\sum_{i\in N} \gamma_i < \sum_{j\in P} T_j$ as Budget Deficit (BD). In CC literature, it is also natural to assume that $\vartheta_j>T_j, \forall j \in P$~\cite{zubrickas2014provision}. That is, there is sufficient interest in each available project's funding. Hence, when there is surplus budget, it is optimal to fund all the projects, i.e., $P^\star = P$.

Further, we assume that agents do not have any additional information about the funding of the public projects. This assumption implies that their belief towards the projects' funding is symmetric. This is a standard assumption in CC literature~\cite{zubrickas2014provision,damlePricai}.

\begin{figure}[t]
    \centering
    \includegraphics[width=0.5\columnwidth]{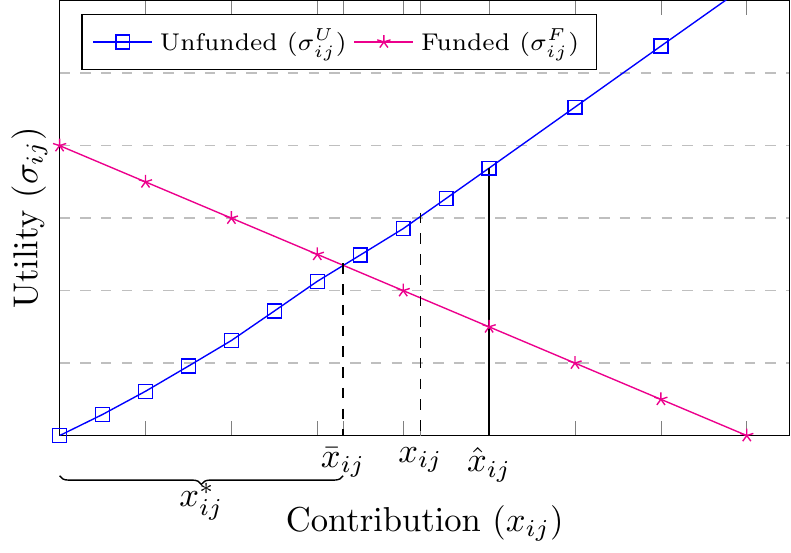}
    \caption{Utility vs. Contribution for agent $i$ for project $j$.}
    \label{fig:my_proof}
\end{figure}
%

\subsection{Single Project Civic Crowdfunding}

{\paragraph{Provision Point Mechanism with Refunds (PPR)}}. For single project CC, i.e., $P=\{1\}$,
\citet{zubrickas2014provision} proposes PPR which employs the following refund scheme $\forall i\in N$,
\begin{equation}\label{eqn::ppr}
 r^{PPR}_{i1}={R}^{PPR}(x_{i1},B_1,C_1)=\left(\frac{x_{i1}}{C_1}\right)B_1.   
\end{equation}
Each agent $i$'s equilibrium contributions are defined such that its funded utility is greater than or equal to its unfunded utility, i.e., $\theta_{i1} - x_{i1}^* \geq  r^{PPR}_{i1} $. We depict such a situation with Figure~\ref{fig:my_proof}. The author shows that the project is funded at equilibrium\footnote{In other words, the introduction of refunds results in the funding of the public project being the unique Nash equilibrium of the game when $\vartheta_1>T_1$. That is, ``bad'' Nash equilibria where the project is not funded are filtered out.} when $\vartheta_1>T_1$, and that it is PSNE for each agent $i$ to contribute $x_{i1}^*$ or the amount left to fund the project, whichever is minimum. More formally,

\begin{theorem}[\cite{zubrickas2014provision}]\label{thm:ppr} 
    In PPR, with $\vartheta_1>T_1$ and $B_1>0$, the set of PSNEs are $\{ x_{i1}^*~|~ x_{i1}^*\leq \bar{x}_{i1},~\forall i; C_1=T_1\}$ if $B_1\leq \vartheta_1 - T_1$, where $\bar{x}_{i1}=\frac{T_1}{B_1+T_1}\theta_{i1}$. Otherwise, the set is empty.
\end{theorem}

We have $\bar{x}_{i1} = \frac{T_1}{B_1+T_1}\theta_{i1}$ as the upper-bound of the equilibrium contribution, $\forall i \in N$. In PPR, the PSNE strategies in Theorem~\ref{thm:ppr} are the unique equilibrium of the game when,

\begin{equation}\label{eqn::equal}
    B_1= \vartheta_1 - T_1 \implies \sum_{i\in N} \bar{x}_{i1} = T_1.
\end{equation}


{\paragraph{Funding Guarantees.}} For single project CC, \citet{damlePricai} show that the project is funded at equilibrium for any refund scheme that satisfies Condition~\ref{cond1}.
Trivially, one may observe that the refund scheme in PPR, i.e., ${r}^{PPR}_{i1}=\left(\frac{x_{i1}}{C_1}\right)B_1$, $\forall i$, also satisfies Condition~\ref{cond1}. \citet{damlePricai} propose other refund schemes which satisfy Condition~\ref{cond1} and are exponential or polynomial in $x$. We remark that our results hold for any refund scheme satisfying Condition~\ref{cond1}.

\section{Funding Guarantees for Combinatorial CC under Budget Surplus}
\label{sec:ee}

For CC under Budget Surplus (Def. \ref{def:bs}) sufficient overall budget exists to fund all the projects. Theorem~\ref{thm::noneqm} shows that despite the sufficient budget, projects may not get funded as the set of equilibrium contributions for an agent may not exist. Unlike Theorem~\ref{thm:ppr}, agents may not have well-defined contributions satisfying PSNE. The non-existence results due to the uneven distribution of budget among the agents. Hence, agents with higher budgets exploit the mechanism to obtain higher refunds while ensuring the projects remain unfunded. 

\begin{theorem}
\label{thm::noneqm}
Given $(R_j)_{j\in P}$ which satisfy Condition~\ref{cond1}, there are Budget Surplus (Def.~\ref{def:bs}) game instances of $\mathcal{M}_{CC}$ with $B_j=\vartheta_j - T_j,\forall j$ such that there is no equilibrium. That is, the set of equilibrium contributions may be empty.
\end{theorem}
\begin{proof} Consider $P$ projects and $N$ agents s.t.  Def.~\ref{def:bs} is satisfied, i.e., $\sum_{i\in N} \gamma_i \ge \sum_{j\in P} T_j$. We can easily construct game instances where there exists non-empty $N_1\subset N$ s.t. $\sum_{i\in N_1} \gamma_i < \min_j T_j$. To satisfy Budget Surplus (Def.~\ref{def:bs}), $N_2 = N\setminus N_1$ must have enough budget so that the agents in $N_1+N_2$ can fund all the projects.

Each agent $i$ receives  a funded utility $\sigma_{ij}^F = \theta_{ij} - x_{ij}$ for contributing $x_{ij}$ towards project $j$. That is, as $x_{ij}\uparrow\implies\sigma_{ij}^F\downarrow$. 
The agent may also receive an unfunded utility of $\sigma_{ij}^U=r_{ij} = R_j(x_{ij},B_j,\cdot)$ for project $j$. Since $R_j$ is monotonically increasing (Condition \ref{cond1}), $x_{ij}\uparrow\implies\sigma_{ij}^U\uparrow$. 
We depict this scenario with Figure \ref{fig:my_proof}. Observe that  $\sigma_{ij}^U$ and $\sigma_{ij}^F$ intersect at the upper-bound of the equilibrium contribution, $\bar{x}_{ij}$ (Theorem~\ref{thm:ppr}), where $\sigma_{ij}^F=\sigma_{ij}^U$. For any $x_{ij} > \bar{x}_{ij}$, $\sigma_{ij}^U>\sigma_{ij}^F$. The rest of the proof (see Appendix~\ref{app::thm2}) shows that $~\exists i\in N_2$ s.t. $\hat{x}_{ij}>\bar{x}_{ij}$ 
which in turn is not possible at equilibrium due to discontinuous utility structure at $\bar{x}_{ij}$.
\end{proof}

Observe that if any project $j$ is funded at equilibrium then the equilibrium set $(x_{i1}^*, \ldots, x_{ip}^*)_{i\in N}$ can not be empty, contradicting Theorem~\ref{thm::noneqm}. Corollary~\ref{cor::bs} captures this observation.

\begin{corollary}\label{cor::bs}
Given $(R_j)_{j\in P}$ satisfying CM (Condition~\ref{cond1}), there are  game instances of $\mathcal{M}_{CC}$ s.t. even with Budget Surplus (Def.~\ref{def:bs}), no project in $P$ may be funded at equilibrium.
\end{corollary}

With~Corollary~\ref{cor::bs}, we prove that Budget Surplus is not sufficient to fund every project at equilibrium. 
To this end, we next identify the sufficient condition to ensure the funding of every project, under Budget Surplus.

\subsection*{Subset Feasibility}
With $N_1$ and $N_2$ in Theorem~\ref{thm::noneqm}'s proof, we assume a specific distribution on agents' budget.
To resolve this, we introduce \emph{Subset Feasibility} which assumes a restriction on each agent's budget distribution. Informally, if each agent $i$ has enough budget to contribute $\bar{x}_{ij}$ (see Figure \ref{fig:my_proof}) for $j \in M$, $M\subseteq P$, then Subset Feasibility is satisfied for $M$. Formally,

\begin{definition}[Subset Feasibility for $M$ (SF$_M$)]\label{def:bf}
Given an instance of $\mathcal{M}_{CC}$ with $(R_j)_{j\in P}$ satisfying Condition~\ref{cond1}, SF$_M$, $M \subseteq P$, is satisfied if, $\forall i\in N $ we have $\gamma_i \geq \sum_{j\in M} \bar{x}_{ij}$. Here,  $\theta_{ij} - \bar{x}_{ij} = R_j(\bar{x}_{ij}, B_j, \cdot)$ (refer Figure \ref{fig:my_proof}).
\end{definition}

\begin{claim}\label{claim1}
Given any $(R_j)_{j\in P}$ whose equilibrium contributions satisfy $\sum_i \bar{x}_{ij} \geq T_j$, we have SF$_P$ $\implies$ BS.
\end{claim}


Claim~\ref{claim1} follows from trivial manipulation (see Appendix~\ref{app::clm1}). Similarly, we have   BS $ \centernot\implies $ SF$_P$.
From Claim~\ref{claim1}, it is welfare optimal to fund every available project under $SF_P$. Theorem~\ref{thm:bf_bs} indeed proves that under SF$_P$ each project $j\in P$ gets funded at equilibrium, thereby generating optimal social welfare. The formal proof is available in Appendix~\ref{app:thm3}.
\begin{theorem}
\label{thm:bf_bs}
Given $\mathcal{M}_{CC}$ and $(R_j)_{j\in P}$ satisfying CM (Condition~\ref{cond1}) such that $SF_{P}$ is satisfied, at equilibrium all the projects are funded, i.e., $C_j = T_j,~\forall j \in P$ if $B_j\leq \vartheta_j-T_j,~\forall j \in P$. Further, the set of PSNEs are:\\ $\left\{(x_{ij}^*)_{j\in P} ~|~ \sigma_{ij}^F(x_{ij}^*;\cdot)\geq \sigma_{ij}^U(x_{ij}^*;\cdot), \forall j\in P,~\forall i\in N\right\}$.
\end{theorem}

Theorem~\ref{thm:bf_bs} implies that, under SF$_P$, the socially efficient equilibrium (with $P^\star = P$) is achieved. Intuitively, under SF$_P$ combinatorial CC collapses to simultaneous single projects; and thus, we can provide closed-form equilibrium contributions. However, SF$_P$ is a strong assumption and, in general, may not be satisfied. In fact, the weaker notion of Budget Surplus itself may not always apply. Therefore, we next study combinatorial CC with Budget Deficit.

\section{Impossibility of Achieving Socially Efficient Equilibrium for Combinatorial CC under Budget Deficit}
\label{sec:we}

We now focus on the scenario when there is \emph{Budget Deficit}, i.e., $\sum_{i\in N} \gamma_i < \sum_{j\in P} T_j$. In this scenario, only a subset of projects can be funded. Unfortunately, identifying the subset of projects funded at equilibrium is challenging. 
In CC, the agents decide which projects to contribute to based on their private valuations and available refund. This circular dependence of the equilibrium contributions and the set of funded projects make providing analytical guarantees challenging on the funded set. To analyze agents' equilibrium behavior and funding guarantees, we fix our focus on the subset of projects that maximize social welfare, i.e., $P^\star$ (Def. \ref{def:wo}).



In this section, we first show that funding $P^\star \subset P$ at equilibrium is, in general, not possible for any $R(\cdot)$ satisfying Condition \ref{cond1}. Second, we prove that even with the stronger assumption of Subset Feasibility of the optimal welfare set, i.e., SF$_{P^\star}$, we may not achieve socially efficient equilibrium due to agents' strategic deviations. Last, we show that computing an agent $i^\prime$'s optimal deviation, given the contributions of the other agents $N\setminus \{i^\prime\}$, is NP-Hard.

\subsection{Welfare Optimality at Equilibrium}

Consider the following example instance.

\begin{instance}\label{eg::one}
Let $P=\{1,2\}$ and $N=\{1,2\}$. Let $\gamma_1=1,\gamma_2=0,\theta_{11}=1,\theta_{12}=2\mbox{~and~}\theta_{22}=1$ with $\theta_{21}=10$.
\end{instance}

In Example~\ref{eg::one}, the maximum funded utility agent 1 can receive from project 1 is $0$ and unfunded utility $r_{11} < \theta_{11} = 1$. On the other hand, the agent obtains a utility of $1$ when contributing to project 2. Hence at equilibrium, project 2 gets funded, although it is welfare optimal to fund project 1. Thus, socially efficient equilibrium is not achieved.

%


\renewcommand{\ALG@name}{Procedure}
\begin{algorithm}[!t] 
\small
\caption{\label{algo::F3} Instance with $P=N=\{1,2\}$ and fixed $R(\cdot)$}
\begin{algorithmic}[1]


 \Procedure{GenerateValues}{$R(\cdot)$}
\State $T_1\leftarrow \mathbb{R}_+$
\State Choose $\theta_{11}$ s.t. $\bar{x}_{11}< T_1 < \theta_{11}$ based on $R_1(\cdot)$

\State Choose $\theta_{21}$ s.t. $\bar{x}_{21}:=T_1 - \bar{x}_{11}$
\State Set $T_2 = \bar{x}_{21}$, $\theta_{12} = 0$ and choose $\theta_{22}$ s.t. 
\Statex ~~~\qquad $\theta_{21} < \theta_{22} < \theta_{11} + \theta_{21} - x_{11}^*$ \Comment{\textcolor{blue}{$P^\star = \{1\}$}}

\State Set $\gamma_1 := \bar{x}_{11}$ and $\gamma_2 := \bar{x}_{21}$ \Comment{\textcolor{blue}{\textit{Satisfying SF$_{P^\star}$}}}

        \State{\Return $\theta$'s, $\gamma$'s, and $T$'s} \Comment{\textcolor{blue}{s.t. Agent 2 deviates}}
  \EndProcedure
%
%
\end{algorithmic}
\end{algorithm}
%

%

Example~\ref{eg::one} is one pathological case where the agent with high valuation has zero budget, leading to sub-optimal outcome at equilibrium. Hence, we next strengthen the assumption on the budgets of the agents. Let $P^\star$ be the non-trivial welfare optimal subset and we assume that Subset Feasibility is satisfied for $P^\star$, i.e., SF$_{P^\star}$. Recall that with SF$_{P^\star}$, we assume that every agent has enough budget to contribute $\bar{x}_{ij}$ in $P^\star$. Theorem~\ref{thm:bf_bd} shows that despite this strong assumption, achieving socially efficient equilibrium may not be possible.


\begin{theorem} 
\label{thm:bf_bd}
Given an instance of $\mathcal{M}_{CC}$, a unique non-trivial $P^\star\subset P$ may not be funded at equilibrium even with Subset Feasibility for $P^\star$, SF$_{P^\star}$, for any set of $(R_j)_{j\in P}$ satisfying Condition \ref{cond1}.
\end{theorem}
\begin{proof} We provide a proof by construction for  $n=2$ and $p=2$ where one of the agents has an incentive to deviate when $P^\star$ is funded. 
Procedure \ref{algo::F3} presents the steps to construct the instance.
We first select the target cost of the project $1$, i.e. $T_1\in\mathbb{R}_+$. Given $(R_j)_{j\in P}$ under Condition~\ref{cond1}, we can always find an $\bar{x}_{11}$ for agent $1$.
At $\bar{x}_{11}$, funded utility is equal to unfunded utility (Figure \ref{fig:my_proof}). Trivially, $\bar{x}_{11} <\theta_{11}$. The rest of the proof shows that the construction defined in Procedure \ref{algo::F3} is always possible for any refund scheme satisfying Condition~\ref{cond1} (see Appendix~\ref{app:thm4}).  
\end{proof}

\subsection{CCC with Budget Deficit: Optimal Strategy}

Theorem \ref{thm:bf_bd} implies that $P^*$ may not be funded at equilibrium even when $P^*$ satisfies Subset Feasibility. In other words, w.l.o.g., an agent $i^\prime$ may have an incentive to deviate from any strategy that funds $P^*$. Motivated by such a deviation, we now address the question: \emph{Given the total contribution by $N\setminus \{i^\prime\}$ agents towards each project $j$, can the agent $i^\prime$ compute its optimal strategy?} We answer this question by (i) showing that such an optimal strategy may not exist if an agent's contribution space is continuous, i.e., $x\in\mathbb{R}_+$, and (ii) if contributions are discretized, then computing the optimal strategy is NP-Hard. 

\smallskip\noindent\textbf{MIP-CC: Mixed Integer Program for CC.} We first describe the general optimization for an agent $i^\prime$ to compute its optimal strategy (i.e., contribution). Assume that the agents $\in N\setminus \{i^\prime\}$ have contributed denoted by $C_{j-i^\prime}$. For agent $i^\prime$'s optimal strategy, we need to maximize its utility given $T_j-C_{j-i^\prime},~\forall j\in P$ and other variables such as the refund scheme $R_j$ and bonus budget $B_j$. Figure~\ref{fig::MIP} presents the formal MIP, namely \emph{MIP-CC}, which follows directly from agent $i^\prime$ utility (Def.~\ref{def:au}).

\begin{figure}[t]
\centering
\fbox{
  \parbox{0.8\columnwidth}{
  \footnotesize
\begin{align*}
 & \max_{(x_{i^\prime j})_{j\in P}} \sum_{j\in P} z_{i^\prime j} \cdot(\theta_{i^\prime j} - x_{i^\prime j}) + (1 - z_{i^\prime j}) \cdot R(x_{i^\prime j},\cdot) && \\
\mbox{\textbf{s.t.}} & \sum_{j\in P} x_{i^\prime j} \leq \gamma_{i^\prime}  && \text{\textcolor{blue}{// Budget Constraint}} \\
& x_{i^\prime j} \leq T_j-C_{j-i^\prime}, \forall j &&  \text{\textcolor{blue}{// Remaining Contribution}} \\
& \begin{rcases}
(x_{i^\prime j}-T_j+C_{j-i^\prime})\cdot z_{i^\prime j} \geq 0 , \forall j 
\\
 x_{i^\prime j}-T_j+C_{j-i^\prime} < z_{i^\prime j}, \forall j \\
  z_{i^\prime j}\in \{0,1\}, \forall j
\end{rcases} && \text{\textcolor{blue}{// Defining Indicator Variable}}
\end{align*}
}
}
\caption{MIP-CC: Mixed Integer Program to calculate Agent $i^\prime$'s optimal strategy given the contributions of the remaining agents $N\setminus \{i^\prime\}$.}
\label{fig::MIP}
\end{figure}

\paragraph{MIP-CC: Optimal Strategy May Not Exist.} We now show that MIP-CC (Figure~\ref{fig::MIP}) may not always admit well-defined contributions.

\begin{instance}\label{ex::2}
Let $P=\{1,2,3\}$ and $N=\{1,2\}$ s.t. both agents are \emph{identical}, i.e.,  each $ i\in N$ has the same value $\theta$ for each $j\in P$ and $\gamma_1=\gamma_2$. Additionally $~\forall j \in P$,  $T_j = T$ and $B_j  = \vartheta - T,$. Let the agents have budget s.t. $\gamma_i = \bar{x}_{i1}~\forall i\in N$, where $\bar{x}_{i1}$ is the upper bound equilibrium contribution for the single project case (see Figure~\ref{fig:my_proof}). 
\end{instance}

\begin{theorem}
Given an instance of $\mathcal{M}_{CC}$ and for any set of $(R_j)_{j\in P}$ satisfying Condition \ref{cond1}, an agent $i^\prime$'s optimal strategy may not exist.
\end{theorem}
\begin{proof}
We use proof by construction. In Appendix~\ref{app::thm5}, we show that for Example~\ref{ex::2}, given agent 1's contribution we can create an instance of $\mathcal{M}_{CC}$ s.t. the set of indicator variable $z$s in MIP-CC can be either $z_1 =\{1,0,0\}$ or $z_2 = \{0,0,0\}$. Then, we show that agent 2's utility becomes $\theta-\gamma_2$ for $z_1$ with strategy $(\gamma_2,0,0)$ and  $2B+R(\gamma_2-\epsilon)$ for $z_2$ with strategy $(\gamma_2-\epsilon,\epsilon/2,\epsilon/2)$ where $\epsilon\leq 0$. As $\epsilon \downarrow$, agent 2's utility for $z_2$ increases. But for $\epsilon=0$, only $z_1$ is possible and agent 2 receives the utility $\theta-\gamma_2$ (strictly less than utility for $z_2$). Due to this \emph{discontinuity} in the utilities, no optimal $\epsilon$ exists, i.e., optimal strategy does not exist.
\end{proof}

\noindent\textbf{MIP-CC-D}. In order to overcome the above non-existence, we discretize the contribution space. More concretely, an agent $i$ can contribute $\kappa\cdot \delta$ where $\kappa\in \mathbb{N}^+$ and $\delta$ the smallest unit of contribution. With this restriction on an agent's contribution, the search space in MIP-CC (Figure~\ref{fig::MIP}) becomes discrete and finite. Consequently, agent $i^\prime$'s optimal strategy always exist. To distinguish MIP-CC with a discrete contribution space, we refer to it as \emph{MIP-CC-D}.  


%

\subsubsection{MIP-CC-D: Finding Optimal Strategy is NP-Hard.} 

We now show that solving MIP-CC for discrete contributions (i.e., MIP-CC-D) is NP-hard. 

\begin{theorem}\label{thm::nphard}
Given an instance of $\mathcal{M}_{CC}$ with discrete contributions and for any set of $(R_j)_{j\in P}$ satisfying Condition \ref{cond1}, computing optimal strategy for agent $i^\prime$, given the contributions of $N\setminus \{i^\prime\}$,  is NP-Hard.
\end{theorem}
\begin{proof}
We divide the proof into two parts. In Part A, we design a MIP tuned for a specific case of combinatorial CC comprising identical projects with a refund scheme satisfying $\sum_j R(x_j,\cdot) = R(\sum_j x_j, \cdot)$. We prove that the MIP is NP-Hard by reducing it from KNAPSACK. Then, in Part B, we show that this MIP reduces to MIP-CC-D. That is, any solution to MIP-CC-D can be used to determine a solution to MIP in polynomial time, implying that MIP-CC-D is also NP-Hard. The formal proof is available in Appendix~\ref{app:nphard}.  
\end{proof}


The following corollary follows from Theorem~\ref{thm::nphard}. We defer the proof to Appendix~\ref{app:cor}.

\begin{corollary}\label{corr:nphard}
Given an instance of $\mathcal{M}_{CC}$ with discrete contributions and for any set of $(R_j)_{j\in P}$ satisfying Condition \ref{cond1}, if all agents except $i^\prime$, $N\setminus \{i^\prime\}$, follow a specific strategy that funds $P^\star\subset P$, then computing the optimal deviation for agent $i^\prime$ is NP-Hard.
\end{corollary}

\section{Experiments}
\label{sec:expt}

%

\smallskip\noindent\textbf{Motivation.} Theorem~\ref{thm:bf_bd} proves that the optimal subset $P^\star$ may not be funded at equilibrium due to agents' strategic deviations. However, computing an agent's optimal deviation is also NP-Hard~(Corollary~\ref{corr:nphard}). These observations highlight that computing closed-form equilibrium strategies in Budget Deficit Combinatorial CC, similar to Theorem~\ref{thm:ppr} and Theorem~\ref{thm:bf_bs}, for agents is challenging. Given this challenge and the hardness of strategic deviations, agents may employ heuristics to increase utility~\cite{zou2010tolerable,lubin2012approximate}. We next propose five heuristics for agents to employ in practice and study their impact on agent utilities and the welfare generated.

\subsection{Heuristics and Performance Measures}

\paragraph{Heuristics.} Given the conflict between agent utilities and $P^\star$'s funding (Theorem~\ref{thm:bf_bd}), we propose the following heuristics for agent $i\in N$, for each project $j\in P$, to employ in practice and observe their utility vs. welfare trade-off.

\begin{enumerate}[leftmargin=*,noitemsep]
    \item \emph{Symmetric}: $x_{ij}=\min(\theta_{ij},\gamma_i/m)$
    \item \emph{Weighted}: $x_{ij}=\left(\frac{\theta_{ij}}{\sum_{k\in P} \theta_{ik}}\right)\gamma_i$
    \item \emph{Greedy-$\theta$}: Greedily contribute $x_{ij} =\bar{x}_{ij}$ in descending order of the projects sorted by $\theta_{ij}, \forall j$
    \item \emph{Greedy-$\vartheta$}:  Greedily contribute $x_{ij} =\bar{x}_{ij}^*$ in descending order of the projects sorted by $\frac{\vartheta_j}{T_j}, \forall j$
    \item \emph{OptWelfare}: $x_{ij} =\bar{x}_{ij}$, $\forall j\in P^\star$ and evenly distribute the remaining budget across $P\setminus P^\star$
\end{enumerate}


\noindent Agents contribute the minimum amount of what is specified by the five heuristics and the amount left to fund the project. We consider OptWelfare as the \textit{baseline} (preferred) heuristic since it generates optimal welfare, i.e., funds $P^\star$.

\paragraph{Performance Measures.} To study the welfare vs. agent utility trade-off, we consider the following performance measures: (i) \emph{Normalized Social Welfare} (SW$_N$) --  Ratio of the welfare obtained and the welfare from $P^\star$ and (ii) \emph{Normalized Agent Utility} (AU$_N$) -- Ratio of the agent utility obtained w.r.t. to the utility when each agent has enough budget to play its PPR contribution $\forall~j\in P$ (see Theorem~\ref{thm:ppr}).

We compare the heuristics when $\alpha\in (0,1]$ fraction of the total agents \emph{deviate}, i.e., choose heuristic $\in$ \{Symmetric, Weighted, Greedy-$\theta$, Greedy-$\vartheta$\}. The remaining $1-\alpha$ fraction of agents use the baseline OptWelfare.


%
\begin{figure}[t]
\centering
\includegraphics[width=\columnwidth]{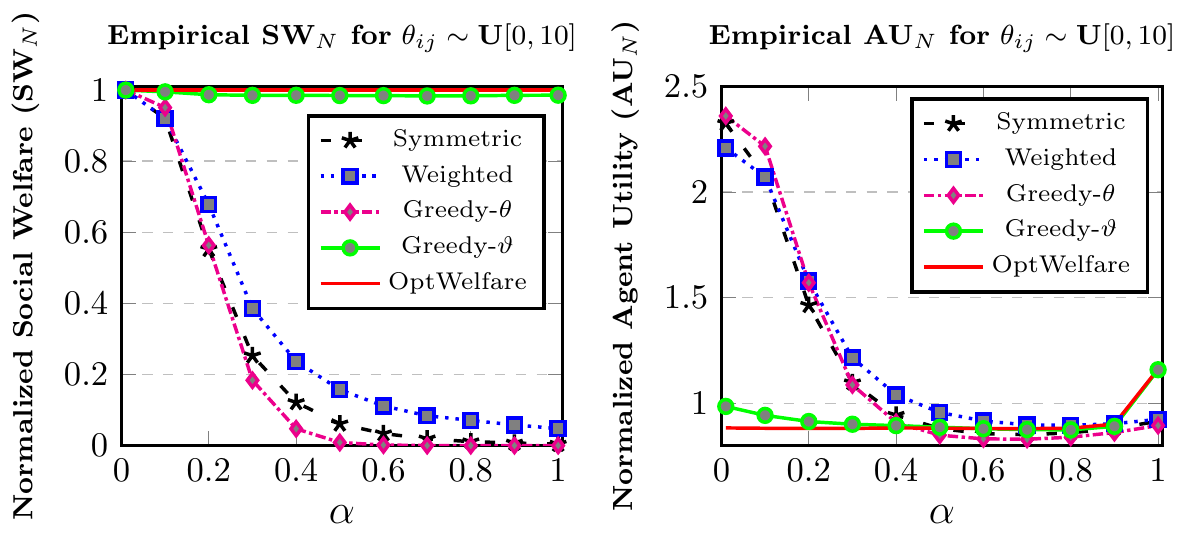}
\caption{Empirical SW$_N$ and AU$_N$ for $\theta_{ij} \sim \mathbf{U}[0,10]$\label{fig:nsw}}
\end{figure}

\begin{figure}[t]
\centering
\includegraphics[width=0.9\columnwidth]{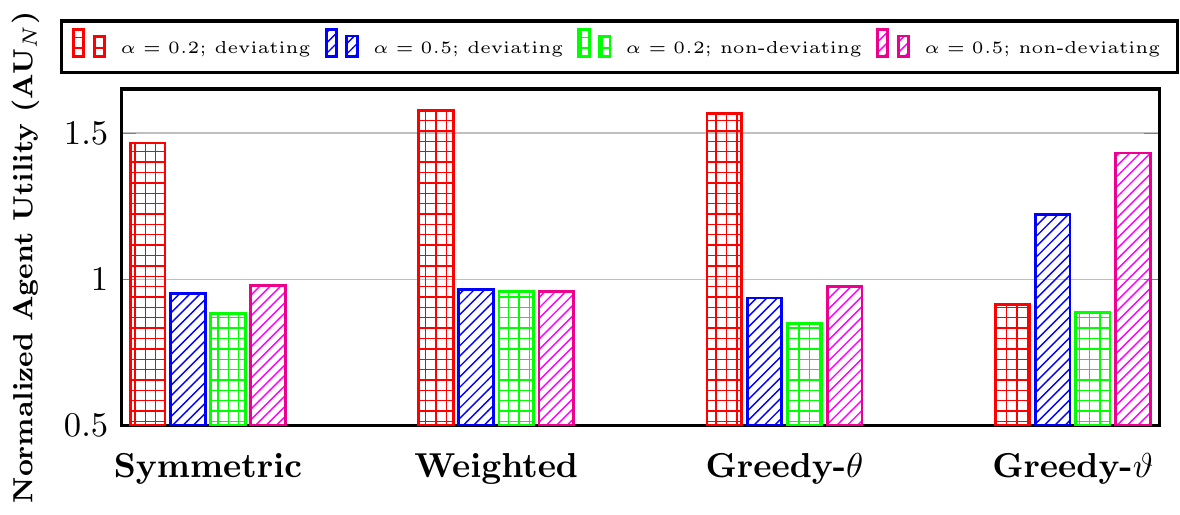}
\caption{AU$_N$ for $\alpha$ fraction of Agents Deviating vs. $1-\alpha$ fraction Non-deviating \label{fig:bar}}
\end{figure}


\subsection{Simulation Setup and Results}

\paragraph{Setup.} We simulate the combinatorial CC game with $n=100$, and $p=10$ and PPR refund scheme (Eq.~\ref{eqn::ppr}).\footnote{The experimental trends presented remain same for different $(n,p)$ pairs, such as $(50,10)$, $(500,10)$, and $(500,20)$.} We sample $\theta_{ij}$s, for each $i\in N$ and $j\in P$, using (i) Uniform Distribution, i.e.,  $\theta_{ij} \sim \mathbf{U}[0,10]$, and (ii) Exponential Distribution, i.e., $\theta_{ij}\sim$ Exp($\lambda=1.5$). Here, $\lambda$ is the rate parameter. When $\theta\sim \mathbf{U}[0,10]$, we get agents whose per-project valuations differ significantly. For $\theta\sim \mbox{Exp}(\lambda=1.5$), the agents have approximately similar per-project valuations.

 We ensure $\vartheta_j > T_j$ and $B_j \in (0, \vartheta_j - T_j]$ for each project $j$ and that the properties Budget Deficit and Subset Feasibility for $P^\star$ are satisfied.   We run each simulation across $100$k instances and  observe the average SW$_N$ and AU$_N$ for each of the five heuristics. We depict our observations with Figures~\ref{fig:nsw} and \ref{fig:bar} when $\theta_{ij} \sim \mathbf{U}[0,10]$. The results for $\theta_{ij}\sim$ Exp($\lambda=1.5$) are presented in Appendix~\ref{app::exp}.

\paragraph{Average SW$_N$ and AU$_N$.} Figure~\ref{fig:nsw} depicts the results when $\theta_{ij} \sim \mathbf{U}[0,10]$. We make three main observations. First, deviating from the baseline heuristic (OptWelfare) is helpful only when few agents deviate, i.e., for smaller values of $\alpha$. Despite such a deviation, we observe that the corresponding decrease in social welfare is marginal. On the other hand, the increase in $\alpha$ reduces the amount of the contributions, and the projects remain unfunded, reducing the social welfare and the agent utilities. Second, deviating from OptWelfare always increases the average agent utility -- at the cost to the overall welfare. Third, Greedy-$\vartheta$ almost mimics OptWelfare, for both SW$_N$ and AU$_N$.


\paragraph{AU$_N$ for Deviating vs. Non-deviating Agents.}
In Figure \ref{fig:bar}, we compare the average utility for the agents who deviate versus those who do not. We let $\alpha=0.2$ fraction of the agents deviate and follow the other four heuristics. From Figure \ref{fig:bar}, we observe that upon deviating to Symmetric, Weighted or Greedy-$\theta$, the $\alpha=0.2$ fraction of agents obtain higher AU$_N$ (red grid bars) compared to the remaining non-deviating agents who do not deviate (green grid bars). In contrast, Greedy-$\vartheta$ shows non-deviation to be beneficial. Since Greedy-$\vartheta$ performs close to OptWelfare, the AU$_N$ for deviation remains low compared to OptWelfare. Crucially, the deviation is not majorly helpful when many agents deviate. When $\alpha = 0.5$, we see comparable average AU$_N$ for agents who deviate and those who do not (blue lined vs. magenta lined bars, respectively). While deviating to Greedy-$\vartheta$ remains non-beneficial.


\paragraph{Discussion and Future Work.} From Figures~\ref{fig:nsw} and \ref{fig:bar}, we see that Greedy-$\vartheta$ performs similar to OptWelfare (which funds $P^\star$). Thus, as the number of projects $p$ increases, to maximize social welfare, it may be beneficial for the agents to adopt Greedy-$\vartheta$ instead of deriving sophisticated strategies based on $P^\star$ (since computing $P^\star$ is NP-Hard).

Generally, it is challenging to determine PSNE contributions for combinatorial CC with budgeted agents. We propose four heuristics and study their welfare and agent utility trade-off. Future work can explore other heuristics that achieve better trade-offs and welfare guarantees. One can also study strategies that perform better on average such as Bayesian Nash Equilibrium. From the experiments, we observe that deviating from OptWelfare may increase agent utility. Thus, one can explore strategies such as $\epsilon$-Nash Equilibrium, which approximates a worst-case $\epsilon$ increase in utility with unilateral deviation. Approximate strategies may also be desirable since finding optimal deviation is NP-Hard. However, the approximation must provide a desirable trade-off between agent utility and welfare.

\section{Conclusion}
\label{sec:conc}

This paper focuses on the funding guarantees of the projects in combinatorial CC. Based on the overall budget, we categorize combinatorial CC into (i) Budget Surplus and (ii) Budget Deficit. First, we prove that Budget Surplus is insufficient to guarantee projects' funding at equilibrium. Introducing the stronger criteria of Subset Feasibility guarantees the projects' funding at equilibrium under Budget Surplus. However, for Budget Deficit, we prove that the optimal welfare subset's funding can not be guaranteed at equilibrium despite Subset Feasibility. Next, we show that computing an agent's optimal strategy (and consequently, its optimal deviation), given the contributions of the other agents, is NP-Hard. Lastly, we propose specific heuristics and observe the empirical trade-off between agent utility and social welfare.

\newpage
\bibliographystyle{unsrtnat}  
\bibliography{references}


\appendix

\newpage
\section{Proofs}

This section presents the formal proofs for the results covered in the main paper.

\subsection{Proof of Theorem 2\label{app::thm2}}
\begin{Theorem}
Given $(R_j)_{j\in P}$ which satisfy Condition~\ref{cond1}, there are Budget Surplus (Def.~\ref{def:bs}) game instances of $\mathcal{M}_{CC}$ with $B_j=\vartheta_j - T_j,\forall j$ such that there is no equilibrium. That is, the set of equilibrium contributions may be empty.
\end{Theorem}
\begin{proof} Consider the set of $P$ projects and $N$ agents s.t.  Def.~\ref{def:bs} is satisfied, i.e., $\sum_{i\in N} \gamma_i \ge \sum_{j\in P} T_j$. 
Denote $N_1\subset N$ as the set of agents whose overall budget is such that even if they contribute their entire budget none of the project gets funded.\footnote{E.g., There exists an agent with budget less than $\min_j T_j$. Also note that, we do not assume $N_1$ to be unique. } This implies $\sum_{i\in N_1} \gamma_i < \min_j T_j$. The remaining set of agents $N_2 = N\setminus N_1$ must have enough budget to fund all the projects (from Def.~\ref{def:bs}).

Note that, each agent $i$ receives  a funded utility $\sigma_{ij}^F = \theta_{ij} - x_{ij}$ for contributing $x_{ij}$ towards project $j$. That is, the funded utility decreases with the increase in contribution. 
The agent may also receive an unfunded utility of $r_{ij} = R_j(x_{ij},B_j,\cdot)$ for project $j$. Since $R_j$ is monotonically increasing (Condition \ref{cond1}), increasing the contribution increases the refund. 
We depict this scenario with Figure \ref{fig:my_proof}. Observe that the plot between funded and unfunded utility w.r.t. contributions meet at the upper-bound of the equilibrium contribution $\bar{x}_{ij}$ (as defined in Theorem~\ref{thm:ppr}), where funded utility is equal to the unfunded. For any contribution $x_{ij} > \bar{x}_{ij}$, the agent receives a greater unfunded utility. 

Consider a scenario where the agents in $N_1$ do not have sufficient budget and thus contribute less than $\bar{x}_{ij}$. Let $C^\prime = \sum_{i\in N_1} \left(\sum_{j\in P} x_{ij}\right)$. Now, we have $x_r = \sum_{j\in P} T_j - C^\prime$ as the remaining amount to fund the set $P$. In order to fund $P$, the remaining agents in $N_2$ may need to contribute $\hat{x}_{ij} > \bar{x}_{ij}$ where $i \in N_2$ (refer Figure~\ref{fig:my_proof}). At $\hat{x}_{ij}$, the unfunded utility is more. Hence, the agents contribute $x_{ij} = \bar{x}_{ij} - \epsilon$, s.t. $\epsilon\to 0$ and $\epsilon > 0$, i.e., the contributions are not well-defined.  
\end{proof}

\subsection{Proof of Claim 1\label{app::clm1}}
\begin{CLAIM}
Given any $(R_j)_{j\in P}$ whose equilibrium contributions satisfy $\sum_i \bar{x}_{ij} \geq T_j$, we have SF$_P$ $\implies$ BS.
\end{CLAIM}
\begin{proof} From the definition of SF$_P$, $\gamma_i \geq \sum_{j\in P} \bar{x}_{ij},~\forall i \in N$. As $\sum_i \bar{x}_{ij} \geq T_j$, we have $\sum_i \gamma_i \geq \sum_{j\in P} \sum_{i} \bar{x}_{ij} \geq \sum_{j\in M} T_j$, i.e., Budget Surplus is also satisfied.  
\end{proof}

\subsection{Proof of Theorem 3\label{app:thm3}}

\begin{Theorem}
Given $\mathcal{M}_{CC}$ and $(R_j)_{j\in P}$ satisfying CM (Condition~\ref{cond1}) such that $SF_{P}$ is satisfied, at equilibrium all the projects are funded, i.e., $C_j = T_j,~\forall j \in P$ if $B_j\leq \vartheta_j-T_j,~\forall j \in P$. Further, the set of PSNEs are:\\ $\left\{(x_{ij}^*)_{j\in P} ~|~ \sigma_{ij}^F(x_{ij}^*;\cdot)\geq \sigma_{ij}^U(x_{ij}^*;\cdot), \forall j\in P,~\forall i\in N\right\}$.
\end{Theorem}
\begin{proof}
For a given refund scheme $R_j$ and corresponding $\bar{x}_{ij}$, each agent $i$'s budget $\gamma_i$ satisfies, using SF$_P$, $\gamma_i \geq \sum_{j\in P}\bar{x}_{ij}$.
Let each agent $i$'s final contribution  $x_{ij}^*\leq \bar{x}_{ij}$; where $\sigma_{ij}^F(\bar{x}_{ij})= \sigma_{ij}^U(\bar{x}_{ij})$.
Since $\sum_{i\in N} \bar{x}_{ij} \geq T_j$, the project gets funded, and each agent obtains a utility of $\sigma_{ij}(x_{ij}^*,x_{-ij}^*;\cdot) = \theta_{ij} - x_{ij}^*$.

W.l.o.g., let us assume that agent $n$ contributes $x_{nj}^* = T_j - \sum_{i \in N\setminus \{n\}} x_{ij}^* \leq \bar{x}_{nj}$ s.t. $\sum_{i\in N} x_{ij}^* = T_j$. 
Then $\forall j\in P$, we can construct the following two deviations for $i$:
\begin{enumerate}[leftmargin=*,noitemsep]
    \item $\hat{x}_{ij} < x_{ij}^*$. For any   $\hat{x}_{ij}$ less than $x_{ij}^*$,  the project $j$ is not funded and $R_j(\hat{x}_{ij}) = \sigma_{ij}(\hat{x}_{ij},x_{-ij}^*;\cdot) < \sigma_{ij}(x_{ij}^*,x_{-ij}^*;\cdot)$.
    \item  $\hat{x}_{ij} > x_{ij}^*$. For any   $\hat{x}_{ij} > x_{ij}^*$, the project is funded but the funded utility reduces, i.e., $\sigma_{ij}(\hat{x}_{ij},x_{-ij}^*;\cdot) = \theta_{ij} - \hat{x}_{ij} < \sigma_{ij}(x_{ij}^*,x_{-ij}^*;\cdot) = \theta_{ij} - x_{ij}^*$.
\end{enumerate}

In both the cases, we obtain the following,
$$ \sigma_{ij}(x_{ij}^*,x_{-ij}^*;\cdot) \geq \sigma_{ij}(\hat{x}_{ij}, x_{-ij}^*;\cdot), \forall \hat{x}_{ij}, \forall j.$$
Therefore, 
$$ \sum_{j\in P} \sigma_{ij}(x_{ij}^*,x_{-ij}^*;\cdot) \geq \sum_{j \in P} \sigma_{ij}(\hat{x}_{ij}, x_{-ij}^*;\cdot), \forall \hat{x}_{ij}$$
From Def.~\ref{def:psne}, $x_{ij}^*$ are PSNE and every project gets funded.  
\end{proof}

\subsection{Proof of Theorem 4\label{app:thm4}}

\begin{Theorem} 
Given an instance of $\mathcal{M}_{CC}$, a unique non-trivial $P^\star\subset P$ may not be funded at equilibrium even with Subset Feasibility for $P^\star$, SF$_{P^\star}$, for any set of $(R_j)_{j\in P}$ satisfying Condition \ref{cond1}.
\end{Theorem}
\begin{proof} We provide a proof by construction for  $n=2$ and $p=2$ where one of the agents has an incentive to deviate when $P^\star$ is funded. 
Procedure \ref{algo::F3} presents the general steps to construct the instance.
Informally, we first select the target cost of the project $1$, i.e. $T_1$, to be any positive real value. Given $(R_j)_{j\in P}$ under Condition~\ref{cond1}, we can always find an $\bar{x}_{11}$ for agent $1$.
At $\bar{x}_{11}$, funded utility is equal to unfunded utility (refer Figure \ref{fig:my_proof}). Trivially, $\bar{x}_{11} <\theta_{11}$. We now prove that the construction defined in Procedure \ref{algo::F3} is always possible.

    \smallskip
    \noindent\underline{Line 3}. We select $\theta_{11}$ such that $x_{11}^* < T_1 < \theta_{11}$. 
    We are required to prove that such a $\theta_{11}$ will always exist. When $\theta_{11} = T_1$, it must be true that, $x_{11}^* < \theta_{11} = T_1$. If $x_{11}^* = T_1$ then $R(x_{11}^*) = 0$ (i.e., funded utility equals unfunded utility), since $R(0)=0$ by construct, contradicting Condition~\ref{cond1}. Hence $x_{11}^* < T_1$, when $\theta_{11} = T_1$. As $x_{11}^*$ is continuous w.r.t. $\theta_{11}$, we can always chose $\theta_{11} > T_1$ s.t. $x_{11}^* < T_1$.
    
    \smallskip
    \noindent\underline{Line 4}. We choose $\theta_{21}$ s.t., $x_{21}^* = T_1 - x_{11}^*$. It is possible since $T_1 - x_{11}^* > 0$.
    
    \smallskip
    \noindent\underline{Line 5}. We let $T_2 = x_{21}^*$ and $\theta_{12} = 0$. Then, select $\theta_{22} < \theta_{11} + \theta_{21} - x_{11}^*$. Observe that, we can always select $\theta_{22} > \theta_{21}$ since $\theta_{11}  - x_{11}^* >0$.
    
    Given the above, and the values $T_2 = T_1 - x_{11}^*$ and $\theta_{12}=0$; we arrive at $\theta_{11} + \theta_{21} - T_1 > \theta_{12} + \theta_{22} - T_2$. Hence, $P^\star = \{1\}$, i.e., the welfare optimal subset is $\{1\}$ and is unique.
    
    \smallskip
    \noindent\underline{Line 6}. To ensure $SF_{P^\star}$, i.e., Subset Feasibility of $P^\star$, we set $\gamma_1 = x_{11}^*$ and $\gamma_2 = x_{12}^*$. Note that, $\gamma_1 + \gamma_2 = T_1 < T_1 + T_2$, Budget Deficit is also satisfied. Thus, Procedure 1 returns valid instance, say $\mathcal{M}^\prime$, of $\mathcal{M}_{CC}$.
    
    Next, note that, in $\mathcal{M}^\prime$, agent 1 contributes $\bar{x}_{11}$ and agent 2 must contribute $\bar{x}_{21}$ to fund project 1 and to achieve socially efficient equilibrium. If agent 2 does not deviate, it will receive an utility $\sigma_{2} = \theta_{21} - \bar{x}_{21}$. 
    However, let agent 2 unilaterally deviate and contribute $\bar{x}_{21}$ to project 2. As $T_2=\bar{x}_{21}$, project 2 gets funded and agent 2 obtains a utility of $\sigma'_2 = \theta_{22} - \bar{x}_{21}$. 
    From Line 5 in Procedure~\ref{algo::F3}, we know that $\sigma'_2 > \sigma_2$, hence agent 2 will deviate, implying $P^\star = \{1\}$ is  not funded at equilibrium for the $\mathcal{M}^\prime$ game. 
\end{proof}

\subsection{Proof of Theorem 5\label{app::thm5}}

\begin{Theorem}
Given an instance of $\mathcal{M}_{CC}$ and for any set of $(R_j)_{j\in P}$ satisfying Condition \ref{cond1}, an agent $i^\prime$'s optimal strategy may not exist.
\end{Theorem}
\begin{proof}
From Example~\ref{ex::2}, let agent 1 contribute its entire budget to project 1, i.e., $T_1-C_1=\bar{x}_{11}$ and $C_2=0$. Now we solve MIP-CC for agent 2. Since $\gamma_i=\bar{x}_{i1},~\forall i$ and $T_j=\bar{x}_{1j}+\bar{x}_{2j},~\forall j$ (see Eq.~\ref{eqn::equal}), the overall budget is enough to fund only a single project. That is, as agent 1 has contributed to project 1, agent 2 can (at best) fund only project 1. So the set of indicator variable $z$s in MIP-CC can be either $z_1 =\{1,0,0\}$ or $z_2 = \{0,0,0\}$.

\smallskip
\noindent\underline{Computing Optimal Strategy}. We have the following cases:

\begin{itemize}[leftmargin=*]
    \item[$\bullet$] $z = \{1,0,0\}$: This is possible when agent 2 contributes $x_{21}=\gamma_2$ to project 1, resulting in project $1$'s funding. The utility agent 2 gets is $\theta - \bar{x}_{21}$.
    \item[$\bullet$] $z = \{0,0,0\}$: Agent 2 may opt to contribute in all the 3 projects in order to grab the maximum refund. As the refund shares from $R_j(\cdot)$ sum exactly to $B_j$, agent 2 can grab the entire bonus from projects 2 and 3 by contributing any arbitrary positive value, $\epsilon >0$. However, agent 2 must contribute $\epsilon$ less in project 1 since, $\gamma_2 = \bar{x}_{21}$. Hence, project 1 becomes unfunded. For optimal strategy, agent 2 must maximize the refund from project 1, i.e., $\underset{x_{21}}{\max} ~R_1(x_{21},\cdot) ~\mbox{s.t.~}  x_{21}< T_1-C_1$. Therefore, $x_{21}$ is not well-defined as $R_1(x_{21},\cdot)$ satisfies CM. So the overall utility by deviating for agent 2 becomes $2B + R(x_{21}-\epsilon,\cdot)$ where an optimal $\epsilon$ is not defined.
\end{itemize}
\end{proof}

\subsection{Proof of Theorem 6 \label{app:nphard}}
\begin{Theorem}
Given an instance of $\mathcal{M}_{CC}$ with discrete contributions and for any set of $(R_j)_{j\in P}$ satisfying Condition \ref{cond1}, computing optimal strategy for agent $i^\prime$, given the contributions of $N\setminus \{i^\prime\}$,  is NP-Hard.
\end{Theorem}
\begin{proof}

For completeness, we first re-write MIP-CC-D (see Figure \ref{fig::MIP}) for discrete contributions. Let $D=\{\kappa\cdot\delta~|~\kappa\in\mathbb{N}\}$ denote the discrete contribution set where $\delta$ is the smallest unit of contribution. Now, the following optimization defines MIP-CC-D for agent $i^\prime$ without the sub-script ``$i^\prime$". Note that, Eq.~\ref{eq::opt_mip} is merely the optimization presented in Figure~\ref{fig::MIP} with the added constraint on the contribution set. 

\begin{align}\label{eq::opt_mip}\tag{A1}
\begin{rcases}
     \max_{(x_{ j})_{j\in P}} &  \sum_{j\in P} z_{j} \cdot(\theta_{ j} - x_{ j}) + (1 - z_{ j}) \cdot R(x_{ j},\cdot) \\
 \mbox{\textbf{s.t.}} & \sum_{j\in P} x_{j} \leq \gamma \\
& x_{ j} \leq T_j-C_{j-i^\prime}, \forall j   \\
& (x_{j}-T_j+C_{j-i^\prime})\cdot z_{j} \geq 0,\  \forall j 
\\
& x_{j}-T_j+C_{j-i^\prime} < z_{ j}, \forall j \\
 & z_{ j}\in \{0,1\}, x_j \in D , \forall j
\end{rcases}
\end{align}

We first construct an optimization problem MIP, which we show is equivalent to the KNAPSACK problem in Part A. In Part B, we show that the MIP reduces to MIP-CC-D implying MIP-CC-D is also NP-Hard.

\subsubsection{Part A: Designing MIP (Eq.~\ref{eqn::newMIP})}
For our proof, we consider a refund scheme $R_j(\cdot)=R(\cdot)$ and bonus budget $B_j=B$ which is the same for each project $j$. $R(\cdot)$ satisfies CM (Condition~\ref{cond1}) such that the refund share satisfies $\sum_j R(x_j,\cdot) = R(\sum_j x_j, \cdot)$.
    \footnote{An example of such a refund scheme can be $R(\cdot)=\frac{x\cdot B_{\min}}{\vartheta_{\max}}$ where $B_{\min}=\min_j (\vartheta_j - T_j) $ and $\vartheta_{\max} =\max_j \vartheta_j$.}
Now we define the MIP as follows. 

\begin{align}\label{eqn::newMIP}\tag{A2}
\begin{rcases}
    & \max_{z=(z_1,\ldots,z_p)}\sum_{j\in P}  z_j (\theta_{j} - r_{j}) + R(\gamma-\sum_{j\in P} z_j r_j,\cdot)\\
    \textbf{s.t.:} & \sum_j z_j r_j \leq \gamma \mbox{~and~}  z_j \in \{0,1\}
\end{rcases}
\end{align}


We now prove that the MIP defined in Eq. \ref{eqn::newMIP} is NP-Hard. The MIP can be re-written as, 
\begin{align*}
\max_{z=(z_1,\ldots,z_p)} & \sum_{j\in P}  z_j (\theta_{j} - r_{j}) + R(\gamma - \sum_{j\in P} z_j r_j,\cdot) \\
=\max_{z=(z_1,\ldots,z_p)} & \sum_{j\in P}  z_j (\theta_{j} - r_{j}) - \sum_{j \in P} z_j R(r_j) + R(\gamma) \\
& \quad \left( \text{By}\ \sum_j R(x_j,\cdot) = R(\sum_j x_j, \cdot),\  z_j \in\{0,1\}. \right) \\
=\max_{z=(z_1,\ldots,z_p)} & \sum_{j\in P}  z_j (\theta_{j} - r_{j} - R(r_j)) + R(\gamma) \\
=\max_{z=(z_1,\ldots,z_p)} & \sum_{j\in P}  z_j (\theta_{j} - r_{j} - R(r_j))
\end{align*}

We reduce the MIP problem from the NP-complete $\operatorname{KNAPSACK}$ problem: given $m$ items with weights $w_1,\cdots, w_m$ and value $s_1,\cdots, s_m$, capacity $B$ and value $V$, does there exist a subset $Q\subseteq \{1,\cdots,m\}$ such that $\sum_{j\in Q} w_j \leq B$ and $\sum_{j\in Q} s_j \geq V$? Given an instance of the  $\operatorname{KNAPSACK}$ problem, we build an instance of the above MIP as follows:
\begin{itemize}
    \item[$\bullet$] The set of projects is the set of items. The amount left for funding the project is $r_j=w_j$. The budget of the agent $i^\prime$ is $\gamma = B$. 
    \item[$\bullet$] The value of each project to the agent is, $\theta_j - r_j - R(r_j)= s_j$, i.e., $\theta_j - R(r_j) = s_j + w_j$.
\end{itemize}
We can see that the utility obtained by choosing a set of projects $Q = \{j | z_j = 1\}$ is exactly equal to the value of choosing set of items $Q$ in the KNAPSACK problem i.e, $\sum_{j\in Q}\theta_j-r_j-R(r_j)=\sum_{j\in Q} s_j $. Also note that the budget constraint is satisfied if and only if the capacity constraint is satisfied. It follows that a solution with value at least $V$ exists in the KNAPSACK problem if and only if there exists a set of projects whose social welfare in the above instance is at least $V$. 

Thus, we reduce MIP from KNAPSACK implying MIP is NP-Hard.
\end{proof}

We next show that MIP in Eq.~\ref{eqn::newMIP} also reduces to MIP-CC-D.

\subsubsection{Part B: Reducing MIP to MIP-CC-D}
Before we show the reduction, we define the following handy notations for the objectives of MIP-CC-D and MIP.
\begin{align*}
    F(z, x) &= \sum_{j\in P} z_{j} \cdot(\theta_{j} - x_{j})  +  \sum_{j\in P} (1 - z_{j}) \cdot R(x_{j},\cdot) && \textcolor{blue}{\mbox{``MIP-CC-D''}} \\
    G(z) &= \sum_{j\in P} z_{j} \cdot(\theta_{j} - r_{j})  +  R(\gamma -\sum_{j\in P} z_j r_j,\cdot) && \textcolor{blue}{\mbox{``MIP in Eq.~\ref{eqn::newMIP}''}}
\end{align*}


Further let, $(Z^F, X^F)$ denote the tuple of feasible solutions to MIP-CC-D and $Z^G$ denote the set of feasible solutions to MIP. 

Given an instance of MIP we construct an instance of MIP-CC-D using the following conditions,

\begin{itemize}
    \item[P1.] For each project $j$, we have $R_j(\cdot)=R(\cdot)$ s.t. $\sum_j R(x_j,\cdot) = R(\sum_j x_j, \cdot)$.
    \item[P2.] Let the remaining contribution be $C_{j-i^\prime} = T_j - r_j,~\forall j \in P$.
\end{itemize}

Using P1 and P2, we now show that any $(z^*, x^*) \in OPT_{MIP-CC-D}$ implies that $z^* \in OPT_{MIP}$.

\begin{Claim}
For the specific setting defined by P1-P2 and given $(Z^F, X^F)$ and $Z^G$ as the tuple of feasible solutions to MIP-CC-D and MIP in Eq.~\ref{eqn::newMIP}, respectively, we have $Z^G \subseteq Z^F$.
\end{Claim}
\begin{proof}
For the proof, we show that given any solution $z\in Z^G$, we can also construct a solution $(z,x)\in (Z^F, X^F)$. That is, given $z$, we can construct $x$ as follows. Let $x_j=r_j$ if $z_j=1$ and $x_j=0$ if $z_j=0$. 

We now have to show that $x$'s construction does not break feasibility of a solution in MIP-CC-D. Observe that, from Eq.~\ref{eq::opt_mip},

\begin{itemize}
    \item Trivially, we have $\sum_j x_j = \sum_j z_jr_j \leq \gamma$ and $x_j \leq r_j = T_j - C_{j-i^\prime}$. That is, the budget constraint and project's funding conditions are satisfied.
    \item If $z_j=0$, then both $(x_j-T_j+C_{j-i^\prime})\cdot z_j \geq 0$ and $x_j-T_j+C_{j-i^\prime}<z_j$ hold.
    \item Likewise, if $z_j=1$ then both $(x_j-T_j+C_{j-i^\prime})\cdot z_j \geq 0$ and $x_j-T_j+C_{j-i^\prime}<z_j$ hold.
\end{itemize}

 Hence for every $z\in Z^G$ we can construct $(z, x)$ s.t., $(z, {x}) \in (Z^F, X^F)$. Hence every solution that is feasible for MIP is also feasible for MIP-CC-D.
\end{proof}



With this, consider the following claim.

\begin{Claim}
\label{claim::partb}

For the specific setting defined by P1-P2 and any fixed $z\in Z^F$, let $\hat{x}_j^z = r_j$ where $z_j = 1$ and $r_j = T_j - C_{j-i^\prime}$ and $\sum_j \hat{x}^z = \gamma$. Then  $\exists \hat{x}^z\in X^F$, such that $F(z, \hat{x}^z) \geq F(z, x)$, $\forall x\in X^F$. 
\end{Claim}
\begin{proof}
 Given a fixed $(z, x)\in (Z^F, X^F)$, observe that,
\begin{align*}
  F(z, x) & = \sum_{j\in P} z_{j} \cdot(\theta_{j} - x_{j})  +  \sum_{j\in P} (1 - z_{j}) \cdot R(x_{j},\cdot) \\
  &  = \sum_{j\in P}  z_j (\theta_{j} - x_{j}) + R(\sum_{j\in P} (1 - z_j) x_j,\cdot) \qquad \left(\text{As}\ \sum_k R(x_k,\cdot) = R(\sum_{k} x_k, \cdot)\right)\\
  & = \sum_{j\in P}  z_j (\theta_{j} - x_{j}) + R(\sum_{j\in P} x_j -\sum_{j\in P} z_j x_j,\cdot)  \\
   & \leq \sum_{j\in P}  z_j (\theta_{j} - x_{j}) + R(\gamma -\sum_{j\in P} z_j x_j,\cdot) \qquad  \left(\text{As}\  \sum_j x_j \leq \gamma \ \text{and}\ R \ \text{satisfies CM}\right) \\
    & \leq \sum_{j\in P}  z_j (\theta_{j} - r_{j}) + R(\gamma -\sum_j z_j x_j,\cdot) \qquad  \left(\text{From~Eq.~}\ref{eq::opt_mip},\  x_j < r_j \implies z_j = 0\right) \\
    & = \sum_{j\in P}  z_j (\theta_{j} - r_{j}) + R(\gamma -\sum_{j\in P} z_j r_j,\cdot) = F(z, \hat{x}^z).  \\
\end{align*}
Hence,  $F(z, \hat{x}^z) =  \sum_{j\in P}  z_j (\theta_{j} - r_{j}) + R(\gamma -\sum_j z_j r_j,\cdot) \ge F(z, x) $ where $\hat{x}_j^z = r_j$ if $z_j = 1$ and $\sum_j \hat{x}_j^z = \gamma$.

To conclude the proof, we now show that the contribution set is also feasible. To this end, observe that from the claim statement we have $\sum_j \hat{x}^z = \gamma$. Further, since for each $z_j=1$, we have by construction $\hat{x}^z_j=r_j\leq T_j-C_{j-i^\prime}$ and for each $z_j=0$ we have $\hat{x}^z_j<r_k$, the second feasibility condition is also met. 
\end{proof}

We can trivially observe that $G(z) = F(z, \hat{x}^z)$ where $\hat{x}_j^z = r_j$ if $z_j = 1$ and $\sum_j \hat{x}_j^z = \gamma$. Let $(z^*, x^*) \in OPT_{MIP-CC-D}$, then $F(z^*, x^*) \geq F(z^*, \hat{x}^{z^*})$ and from Claim~\ref{claim::partb}, $F(z^*, \hat{x}^{z^*}) \geq F(z^*, x^*)$, hence $x^* = \hat{x}^{z^*}$. That is, $$G(z^*) = F(z^*, x^*) = \max_{z \in Z^F} F(z, \hat{x}^{z^*}) = \max_{z\in Z^F} G(z).$$

Given that (i) the feasibility set of MIP is a subset of the feasibility set of MIP-CC-D and (ii) $z^*$ is also feasible to MIP since $\sum_j x_j^* \leq \gamma$ and $x_j^* \leq r_j, \ \forall j$, we have $\sum_j z_j^* r_j \leq \gamma$.
Hence $z^*$ is also optimal to MIP.

In summary, Part B shows that MIP in Eq.~\ref{eqn::newMIP} reduces to MIP-CC-D. That is, any solution to MIP-CC-D can be used to determine a solution to MIP in Eq.~\ref{eqn::newMIP}. Hence, MIP-CC-D is also NP-Hard.

\subsection{Proof of Corollary 2\label{app:cor}}
\begin{Corollary}
Given an instance of $\mathcal{M}_{CC}$ with discrete contributions and for any set of $(R_j)_{j\in P}$ satisfying Condition \ref{cond1}, if all agents except $i^\prime$, $N\setminus \{i^\prime\}$, follow a specific strategy that funds $P^\star\subset P$, then computing the optimal deviation for agent $i^\prime$ is NP-Hard.
\end{Corollary}
\begin{proof}
Let $N\setminus \{i^\prime\}$ agents follow a specific strategy and contribute to the projects s.t. $P^\star$ is not funded. For any such strategy, given $N\setminus \{i^\prime\}$ contributions, agent $i^\prime$'s optimal deviation is again given by MIP-CC-D. Hence NP-Hardness follows directly from Theorem~\ref{thm::nphard}.  
\end{proof}

\section{Illustration of Procedure~\ref{algo::F3} for PPR\label{app:illus}} 
To further clarify the impossibility presented in Theorem~\ref{thm:bf_bd}, we demonstrate Procedure \ref{algo::F3} when $(R_j)_{j\in P}$ is the PPR refund scheme (Eq.~\ref{eqn::ppr}).
As required, consider a setting with $P=\{1,2\}$ and $N=\{1,2\}$. Let $T_1=10$ and $B_1=1$. We choose $\theta_{11}=10.9$ which implies $\bar{x}_{11}=9.91$ using Theorem~\ref{thm:ppr}. That is,  $\bar{x}_{11}< T_1 < \theta_{11}$. Now, fix $\bar{x}_{21}=T_1-\bar{x}_{11}=0.99$. This also implies $T_2=0.99$.  The upper bound on the equilibrium contribution, $\bar{x}_{21}=0.99$, is possible for the value $\theta_{21}=1.089$. Select $\theta_{22}=1.9$, to get $\theta_{22}>\theta_{21}$. This also implies $P^\star = \{1\}$, i.e., $\theta_{11}+\theta_{21}-T_1=1.989$ and $\theta_{12}+\theta_{22}-T_2=0.91$. Despite this, agent 2 will deviate since $\theta_{22}-T_2>\theta_{21}-\bar{x}_{21}$.

\section{Exponential Distribution\label{app::exp}}



\begin{figure}[t]
\centering
\includegraphics[width=\columnwidth]{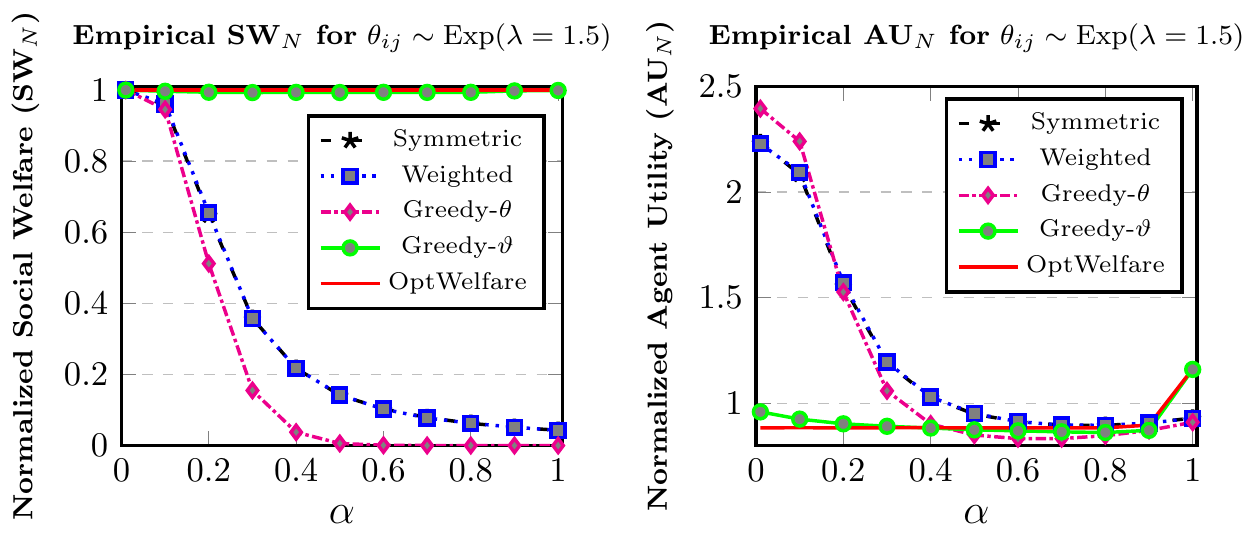}
\caption{Empirical SW$_N$ and AU$_N$ for $\theta_{ij} \sim \mbox{Exp}(\lambda=1.5)$\label{fig:nsw-exp}}
\end{figure}

We now provide the results when agents sample their valuations using the Exponential distribution, i.e., for each agent $i$ and for the project $j$, let $\theta_{ij} \sim \mathbf{Exp}(\lambda=1.5)$. Here, $\lambda$ is the rate parameter. Figure~\ref{fig:nsw-exp} depicts the empirical Normalized Social Welfare (SW$_N$) and Normalized Agent Utility (AU$_N$).

Despite the difference in sampling $\theta$s, we observe that the overall trend in AU$_N$ and SW$_N$ is the same as that reported for the Uniform case (in the main paper). That is, agent deviation is helpful only for smaller $\alpha$s. And that such a deviation results in only a marginal decrease in the overall welfare obtained.

 These results further suggest that agents may prefer the strategy Greedy-$\vartheta$ as the number of projects increases. One can also study the notion of $\epsilon$-Nash Equilibrium for these cases and provide theoretical bounds on the tolerance $\epsilon$ vis-a-vis the social welfare.

\end{document}